\documentclass[a4paper,USenglish]{lipics-v2016}
 
\usepackage{microtype}
\usepackage[ruled,linesnumbered,vlined]{algorithm2e}
\usepackage{cite}

\graphicspath{{./graphics/}}

\title{Sliding Suffix Tree}
\titlerunning{Sliding Suffix Tree} 

\author[1,2]{Andrej Brodnik}
\author[1]{Matevž Jekovec}
\affil[1]{University of Ljubljana, Faculty of Computer and Information Science, SI\\
	\texttt{\{andrej.brodnik,matevz.jekovec\}@fri.uni-lj.si}}
\affil[2]{University of Primorska, Faculty of Mathematics, Natural Sciences and Information Technologies, SI\\
	\texttt{andrej.brodnik@upr.si}}
\authorrunning{A. Brodnik and M. Jekovec} 

\Copyright{Andrej Brodnik and Matevž Jekovec}

\subjclass{F.2.2 Nonnumerical Algorithms and Problems, E.1 DATA STRUCTURES, F.1.2 Modes of Computation}
\keywords{suffix tree, online pattern matching, sliding window, lowest common ancestor}

\EventEditors{John Q. Open and Joan R. Acces}
\EventNoEds{2}
\EventLongTitle{42nd Conference on Very Important Topics (CVIT 2016)}
\EventShortTitle{CVIT 2016}
\EventAcronym{CVIT}
\EventYear{2016}
\EventDate{December 24--27, 2016}
\EventLocation{Little Whinging, United Kingdom}
\EventLogo{}
\SeriesVolume{42}
\ArticleNo{23}

\begin{document}

\maketitle

\begin{abstract}
We consider a sliding window $W$ over a stream of characters from some alphabet of constant size. The user wants to perform deterministic substring matching on the current sliding window content and obtain positions of the matches. We present an indexed version of the sliding window based on a suffix tree. The data structure of size $\Theta(|W|)$ has optimal time queries $\Theta(m+occ)$ and amortized constant time updates, where $m$ is the length of the query string and $occ$ is the number of its occurrences.
\end{abstract}

\section{Introduction and Related Work}

Text indexing and big data in general is a well studied computer science and engineering field. A specially intriguing area is (infinite) streams of data which are too big to fit onto disk, and consequently, cannot be indexed in the traditional way (e.g.\ by using FM-index \cite{Ferragina2000}). In practice, data streams are processed on-the-fly by efficient, carefully engineered filters. An excerpt of the data called text features are stored for later usage, while the original stream data is discarded.

In our research, we consider an infinite stream of characters where the main memory holds the most recent characters in the stream in terms of sliding window. At any moment, a user wants to find all occurrences of the given substring in the current window. In general, to answer the query, we could construct an automaton from a query using KMP \cite{Knuth1977} or Boyer-Moore \cite{Boyer1977}, and then feed the stream to the constructed automaton. This however requires that all queries are known in advance. On the other hand, if the query arrives on-the-fly, the automaton needs to be constructed from scratch. In both cases we need to scan the whole window which requires linear time in the size of a window. A better possibility would be to run Ukkonen's online suffix tree construction algorithm \cite{Ukkonen1995} and construct the suffix tree. When the query arrives, we inject a delimiter character to finalize the suffix tree construction and perform the query on the constructed tree. However, finalizing might take, in the worst case, linear time in the size of the window.

In this paper we show how to construct and maintain an indexed version of the sliding window allowing a user to find occurrences of a substring in optimal time and space. This is the first data structure for on-the-fly text indexing which requires amortized constant time for updates and worst case optimal time for queries.

In the following section we define the notation and preliminary data structures and algorithms. In Section \ref{sec:sliding-suffix-tree} we formally present a sliding suffix tree and we conclude in Section \ref{sec:conclusions} with discussion and open problems.

\section{Notation and Preliminaries}

Capital letters $A, B, C \ldots$ denote strings and lower case letters $i, j \ldots$ integers except for $c$ which denotes an arbitrary character. Further, lower case Greek letters $\alpha, \beta \ldots$ represent nodes in a tree, and calligraphic capital letters $\mathcal{T}, \mathcal{L}$, and $\mathcal{A}$ tree-based data structures. We denote concatenation of two strings by simply writing one string beside the other, e.g.\ $AB$ and the length of a string $A$ as $|A|$. By $A[i:j]$ we denote a substring of $A$ starting at position $i$ and ending at $j$ inclusive, where $1 \leq i \leq j \leq |A|$. Suffix of $A$ starting at $i$ is $A[i:]$ and a prefix of $A$ ending at $j$ is $A[:j]$ both inclusive.

We denote by $W$ the sliding window over an infinite input stream of characters which are from an alphabet of constant size. By $n$, we denote the number of all characters read so far. To store a suffix starting at the current position, we store the current $n$. At any later time $n'$, we can retrieve the content of this suffix as $W[n-(n'-|W|):]$, where $n-n'<|W|$, for the suffix to be present in $W$.

\subsection{Suffix Tree, Suffix Links Tree, and Lowest Common Ancestor}

A \emph{suffix tree} is a dictionary containing each suffix of the text as a key with its position in the text as its value. The data structure is implemented as a PATRICIA tree, where each internal node stores a skip value and the first (discriminative) character of the incoming edge, whereas each leaf stores the position of the suffix in the text. We denote by $|\alpha|$ a \emph{string depth} operation of some node $\alpha$ in a suffix tree and define it as a sum of all skip values from the root to $\alpha$. Each edge implicitly contains a \emph{label} which is a substring of a suffix starting and ending at the string depth of the originating and the string depth of the terminating node respectively. We say a node $\alpha$ \emph{spells out} string $A$, where $A$ is a concatenation of all labels from the root to $\alpha$. Or more formally, take a leaf in a subtree of $\alpha$ and let it store position $i$ of a suffix, then $A=W[i-(n-|W|):i+|\alpha|-(n-|W|)]$. Next, let $A$ and $A'$ be strings which are spelled out by nodes $\alpha$ and $\alpha'$ respectively. We define a \emph{suffix link} as an edge from node $\alpha'$ to $\alpha$, if $A=A'[2:]$, and denote this by $\alpha = \text{suffix\_link}(\alpha')$. If we follow suffix links from $\alpha'$ $i$ times, we write this as $\text{suffix\_link}(\alpha')^i$.

We define a \emph{suffix links tree} as follows. For each internal node in a suffix tree let there be a node in a suffix links tree. For each suffix link from $\alpha'$ to $\alpha$ in the suffix tree, $\alpha$ is a parent of $\alpha'$ in the suffix links tree. Consequently, following a suffix link from $\alpha'$ $i$ times is the same as finding the $i^{th}$ ancestor of $\alpha'$ in the suffix links tree.

The lowest common ancestor (LCA) of nodes $\alpha$ and $\alpha'$ in a tree is the deepest node which is an ancestor of both $\alpha$ and $\alpha'$. The first constant time, linear space LCA algorithm was presented in \cite{Harel1984} and later simplified by \cite{Schieber1988}. The dynamic version of the data structure still running in constant time and linear space was introduced in \cite{Cole2005}. We will use this result to perform constant time LCA lookups and maintain it in amortized constant time.

\subsection{Ukkonen's online suffix tree construction algorithm}
\label{chap:ukkonen}

In \cite{Ukkonen1995} Ukkonen presented a suffix tree construction algorithm which builds the data structure in a single pass. During the construction, the algorithm maintains the following invariants in amortized constant time:
\begin{itemize}
	\item implicit buffer $B$ which corresponds to the longest repeated \emph{suffix} of the text processed so far,
	\item the \emph{active node} $\beta$ which represents a node where we end up by navigating $B$ in the suffix tree constructed so far, i.e.\ $B$ is a prefix of a string spelled out by $\beta$.
\end{itemize}

The execution of the algorithm can be viewed as an automaton with two states. In the \emph{buffering state}, the automaton reads a character $c$ from the input stream and implicitly extends all labels of leaves by $c$. Then, it checks whether $c$ matches the next character of the prefix of length $|B|+1$ spelled out by $\beta$. If it does, $c$ is appended to $B$ and the automaton remains in the buffering state reading the next character. When $|B|=|\beta|$, a child in direction of $c$ becomes a new active node.

On the other hand, if character $c$ does not match, automaton switches to the \emph{expanding state}. First it inserts a new branch in direction of $c$ with a single leaf storing the suffix position $n-|B|$. If $|B|=|\beta|$, the new branch is added as a child to $\beta$. Otherwise, if $|B|<|\beta|$ the incoming edge of $\beta$ is split such that the string depth of the newly inserted internal node is $|B|$, and the new branch is added to this node. Once the branch is inserted, the first character is removed from $B$ obtaining new $B'=B[2:]$. A new active node corresponding to $B'$ is found in the following way. Let $\alpha$ denote the parent of the original active node $\beta$. Then the new active node $\beta'$ is a node obtained by navigating suffix $B'[|\alpha|:]$ from a node $\text{suffix\_link}(\alpha)$. When $\beta'$ is obtained, $c$ is reconsidered. If a branch in direction of $c$ exists, the automaton switches to buffering state. Otherwise, it remains in the expanding state and repeats the new branch insertion. Each time the expanding state is re-entered, $B$ is shortened for one character. In the worst case, if $c$ does not occur in the text yet, the suffix links will be followed all the way up to the root node, and $c$ will be added as a new child to the root node. In this case the implicit buffer $B$ will be an empty string.

We say the currently constructed suffix tree is \emph{unfinalized}, until $B$ is completely emptied. Moreover, there are exactly $|B|$ leaves missing in the unfinalized tree and these correspond to suffixes of $B$. For finite texts we finalize the suffix tree at the end by appending a unique character $\$$ which forces the algorithm to empty $B$ and finalize the tree. For infinite streams however, there is no final character. Consequently, we need to support:

\begin{enumerate}
	\item Queries: When performing queries, we need to report the occurrences both in the partially constructed suffix tree and in $B$.
	\item Maintenance: The original Ukkonen's algorithm supports adding a new character to the indexed text. When a window is shifted, we also remove the oldest (longest) suffix from the text.
\end{enumerate}


\section{Sliding Suffix Tree}
\label{sec:sliding-suffix-tree}

The \emph{sliding suffix tree} is an indexed version of the current sliding window content $W$. Formally, we define two operations:
\begin{itemize}
	\item \texttt{find($W$, $Q$)} --- returns all positions of the query string $Q$ in $W$.
	\item \texttt{shift($W$, $c$)} --- appends a character $c$ to $W$ and removes the oldest character from $W$.
\end{itemize}

Initially, $W$ is empty and until the length of $W$ reaches the desired size, \texttt{shift} operation only appends new characters.

The sliding suffix tree is built on top of Ukkonen's online suffix tree construction algorithm. We maintain a possibly unfinalized suffix tree $\mathcal{T}$ including implicit buffer $B$ and active node $\beta$ (Fig. \ref{fig:T-W-B} on the left). Figure \ref{fig:T-W-B} on the right illustrates the position of $W$ and $B$ in a stream. Notice $B$ is always a proper suffix of $W$. Additionally, we maintain a suffix links tree of $\mathcal{T}$, $\mathcal{L}$, with auxiliary data structure $\mathcal{A}$ required for constant time LCA on $\mathcal{L}$.

\begin{figure}[htb]
	\begin{center}
		\includegraphics[width=14cm]{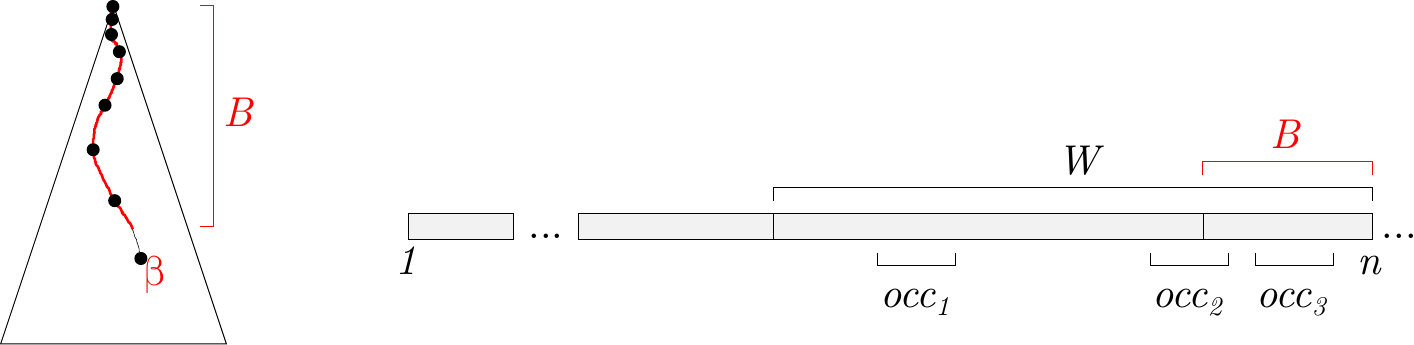}
	\end{center}
	\caption{On the left: Illustration of partially constructed suffix tree $\mathcal{T}$ with implicit buffer $B$ and active node $\beta$. On the right: Illustration of the stream, the sliding window $W$, the implicit buffer $B$, and three cases for positions of the query strings $occ_1$, $occ_2$, and $occ_3$.}
	\label{fig:T-W-B}
\end{figure}

In the next two subsections we show how to perform the find operation in time $\Theta(|Q|+occ)$ in the worst case and the shift operation in constant amortized time. As a model of computation, we use the standard RAM model.

\subsection{Queries}

To find all occurrences of query $Q$ in $W$, we first navigate $Q$ in $\mathcal{T}$. Let $\mathcal{T}_Q$ correspond to a subtree rooted at the node at which we finished the navigation. Leaves of $\mathcal{T}_Q$ make up the first part of the resulting set. In Figure \ref{fig:T-W-B} $occ_1$ corresponds to such occurrence. Also, position of $occ_2$ in the same figure will be contained in one of the leaves of $\mathcal{T}_Q$, since $\mathcal{T}$ contains all suffixes that start at the beginning of $W$ up to the beginning of $B$.

The second part of the resulting set are the missing leaves of $\mathcal{T}_Q$ due to the unfinalized state of $\mathcal{T}$. Intuitively, these leaves correspond to suffixes of $B$ which start with $Q$. $occ_3$ in Figure \ref{fig:T-W-B} illustrates one such position. Obviously, if $|B|<|Q|$ there are no matches of $Q$ in $B$ and we solely return the leaves of $\mathcal{T}_Q$. If $|B|=|Q|$, we test whether the active node $\beta$ is the same node as the root of $\mathcal{T}_Q$. If it is, we add one additional occurrence at position $n-|B|$ to the resulting set.

The case $|B|>|Q|$ requires special attention. One solution would be to scan $B$ for $Q$ using KMP or similar approaches. But since $|B|=O(|W|)$ in the worst case, we cannot afford the scan. In the remainder of this subsection we show how to determine the missing leaves in time $O(|Q|+occ)$. First, we claim that the navigated subtree $\mathcal{T}_Q$ always exists, if there are any occurrences of $Q$ to be found in $B$.

\begin{lemma}
	If $Q$ exists in buffer $B$, then a subtree $\mathcal{T}_Q$ exists by navigating the query $Q$ in $\mathcal{T}$.
\end{lemma}
\begin{proof}[Proof]
	If $Q$ exists somewhere in $B$, then $Q$ is a substring of a string spelled out by $\beta$.
	From the property of the suffix tree, by following the suffix links from $\beta$ we will find a node which spells out a string with $Q$ at the beginning. This node is a root of $\mathcal{T}_Q$.	
\end{proof}

To consider occurrences of $Q$ in $B$ where $|B|>|Q|$, we determine the relation of each node in $\mathcal{T}_Q$ to $\beta$. Since $|\mathcal{T}_Q| = O(occ)$ we can afford this operation, if we spend at most constant time per node. We proceed depending on whether $\beta$ is an internal node of $\mathcal{T}$ or not.

\begin{lemma}
	\label{lemma:beta_tq}
	Let $\beta$ be the active node of $\mathcal{T}$, and let $\beta$ be an internal node. String $Q$ is located in $B$ at position $i$, iff $\text{suffix\_link}(\beta)^i$ is a node in $\mathcal{T}_Q$.
\end{lemma}
\begin{proof}
	($\Rightarrow$) We need to prove that a node corresponding to a suffix of $B$ which starts with $Q$ exists in $\mathcal{T}_Q$ since $\mathcal{T}$ is not finalized. Recall the expanding state of Ukkonen's algorithm. At each call, the operation adds a leaf and possibly an internal node, whereas the existing internal nodes are left untouched. Since $\beta$ is an internal node, no changes will be made either to it or the nodes visited when recursively following the suffix link from $\beta$, since they are also internal nodes. Therefore, a node corresponding to a suffix of $B$ which begins with $Q$ exists in $\mathcal{T}_Q$, if such a suffix exists in $B$.
	
	($\Leftarrow$) By definition of $\beta$, $B$ is a prefix of a string which $\beta$ spells out. $\beta$ is also an internal node, so it will always contain an outgoing suffix link (in case $|\beta|=1$, let the suffix link point to the root node). When following the suffix link of $\beta$, each time we implicitly remove one character from the beginning of $B$. Suppose we follow the suffix link $i$ times and reach a node which is a member of $\mathcal{T}_Q$. By definition of the suffix tree, each node in $\mathcal{T}_Q$ spells out a string which starts with $Q$. Therefore, our reached node corresponds to a suffix of $B$ at position $i$ and starts with $Q$.
\end{proof}

By using the lowest common ancestor operation (LCA) we can check in constant time whether a node is reachable from another node by following the suffix links in $\mathcal{T}$. If $\alpha$ is an ancestor of $\beta$ in $\mathcal{L}$ (i.e.\ the LCA of $\alpha$ and $\beta$ in $\mathcal{L}$ is $\alpha$), then $\alpha$ is reachable by following the suffix links from $\beta$ in $\mathcal{T}$. To determine all occurrences of $Q$ in $B$ in time $O(|\mathcal{T}_Q|)$, for each candidate node $\alpha$ in $\mathcal{T}_Q$ we find its LCA with $\beta$ in $\mathcal{L}$. If the LCA is $\alpha$ and $\alpha$ is an $i^{th}$ ancestor of $\beta$ in $\mathcal{L}$, then by Lemma \ref{lemma:beta_tq} $Q$ is located in $B$ at position $i$.

If $\beta$ is a leaf of $\mathcal{T}$, we cannot use the approach described above, because leaves do not have usable suffix links. We find occurrences of $Q$ in $B$ by exposing a repetitive pattern $P$ inside $B$.

\begin{lemma}[The Buffer Pumping Lemma]
	The buffer $B$ is extended by a new character $c$ during Ukkonen's suffix tree construction algorithm without inflicting the expansion of a tree, iff $c$ corresponds to the next character in a repetitive pattern $P$ inside $B$.
\end{lemma}
\begin{proof}
	($\Rightarrow$)
	Since $c$ does not inflict the expansion of a tree, $Bc$ occurred in the text before. Let $x$ denote the position of last such occurrence as illustrated on Figure \ref{fig:pattern-in-B}. Notice that $B$ starting at $x$ and $B$ starting at $n-|B|$ overlap. Consequently, character $c=W[x-(n-|W|)+|B|+1]$ and in turn $B$ is a concatenation of patterns $P$, $B=P^kP'$ where $k = \lfloor \frac{|B|}{|P|} \rfloor$ and the last repetition $P'$ might be empty.
	
	($\Leftarrow$)
	Given $B$ and a repetitive pattern $P$, we can extend $B$ by a new character $c=P[(|B|+1 \bmod |P|) + 1]$. The expansion of the tree will not occur because $Bc$ was present in the text before and consequently a corresponding edge in partially constructed suffix tree will exist.
\end{proof}

\begin{corollary}
	Let there be a single leaf in the subtree obtained when navigating $B$ in $\mathcal{T}$ and let $x$ be the position stored in this leaf. The repetitive pattern $P$ inside $B$ is $W[x-(n-|W|):n-|B|]$.
\end{corollary}
\begin{proof}
	Since the leaf storing $x$ is the only leaf in the obtained subtree, there are exactly two occurrences of $B$ in the text. The first one at position $x$ and the second one at position $n-|B|$. If $|P|>|B|$, then $B$ is a prefix of $P$, because the leaf spelling out $P$ was obtained by navigating $B$. If $|P|<|B|$, then $B=P^k P'$ due to the buffer pumping lemma.
\end{proof}

\begin{figure}[htb]
	\begin{center}
		\includegraphics[width=9.5cm]{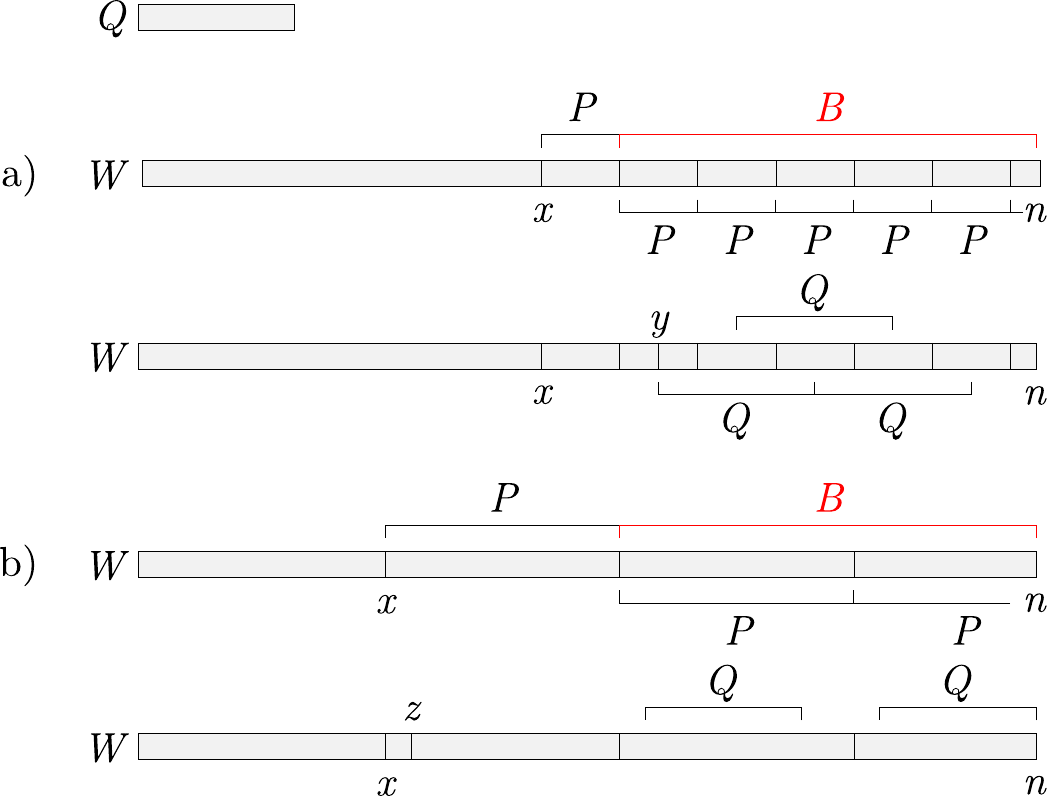}
	\end{center}
	\caption{Structure of $B$ relative to $P$ when $\beta$ is a leaf. Subfigures a) and b) illustrate cases for $|P|\leq|Q|$ and $|P|>|Q|$ respectively. Below each subfigure is an illustration of query $Q$ relative to $y$ and $z$ respectively.}
	\label{fig:pattern-in-B}
\end{figure}

With the help of the lemma and the corollary above we can efficiently determine the positions of $Q$ in $B$ by exposing the repetitive property of the pattern $P$ inside $B$. Depending on the length $|P|$, two cases are possible as illustrated in Figure \ref{fig:pattern-in-B}. If $|P| \leq |Q|$ (Fig.\ \ref{fig:pattern-in-B}.a), we scan for $Q$ in $B$ up to position $2|Q|-1$ inside $B$ and for each such occurrence of $Q$ at some position $y$ we add occurrences $y, y+|P|, y+2|P|, \ldots$ to the resulting set until we reach $n-|Q|$. We require $O(|Q|+occ)$ time in the worst case. If $|P|>|Q|$ (Fig.\ \ref{fig:pattern-in-B}.b), we visit the leaves of $\mathcal{T}_Q$ and consider the suffixes starting inside the interval $x:n-|B|-1$ of the stream. For each such occurrence $z$, we add $z+|P|, z+2|P|\ldots$ to the resulting set until we reach $n-|Q|$. We spend $O(occ)$ time in the worst case.

The data structure we used consists of $(\mathcal{T}, \mathcal{L}, \mathcal{A})$, where $\mathcal{T}$ requires $O(|W|)$ space in the worst case (i.e.\ $|B|=0$) and assuming an alphabet of constant size. Next, $\mathcal{L}$ contains the same number of nodes as $\mathcal{T}$ and is oblivious to the alphabet size, so the space complexity has the same upper bound. Finally, $\mathcal{A}$ used for constant time LCA queries on $\mathcal{L}$ requires linear space in terms of the number of nodes in $\mathcal{L}$. This brings us to the following theorem.

\begin{theorem}
	A user can find all occurrences of query $Q$ in a sliding suffix tree of size $O(|W|)$ in time $\Theta(|Q|+occ)$.
\end{theorem}

\subsection{Maintenance}

To shift window $W$, we read a character $c$ and add it to our data structure and at the same time remove the oldest (longest) stored suffix. During the maintenance no queries can be performed.

To add a character, we first execute the original Ukkonen algorithm as described in subsection \ref{chap:ukkonen}. During the expanding state we add to $\mathcal{T}$ either one node (a new leaf is added to the active node) or two nodes (the incoming edge of the active node is split and a new leaf is added). Since $\mathcal{L}$ contains only internal nodes of $\mathcal{T}$, it remains unchanged in the first case and in the second case, a node is also added to $\mathcal{L}$ as follows.

When the expanding state is visited the first time, a new internal node $\gamma'_\mathcal{T}$ is added to $\mathcal{T}$. We also add a new node $\gamma'_\mathcal{L}$ to $\mathcal{L}$. At this point no suffix link originating in $\gamma'_\mathcal{T}$ has been set, so $\gamma'_\mathcal{L}$ does not have a parent in $\mathcal{L}$ yet. In the next step either an expanding state is re-entered or a buffering state is entered. If the expanding state is re-entered, we repeat the procedure obtaining new nodes $\gamma_\mathcal{T}$ and $\gamma_\mathcal{L}$. Now, a suffix link is created from $\gamma'_\mathcal{T}$ to $\gamma_\mathcal{T}$ and consequently a parent of $\gamma'_\mathcal{L}$ becomes $\gamma_\mathcal{L}$. If the buffering state is entered, either a root node or a node containing the matched character is reached. Instead of creating new nodes in $\mathcal{T}$ and $\mathcal{L}$ as we did in the expanding state, we create a suffix link to an existing node in $\mathcal{T}$ and set the parent of a node in $\mathcal{L}$ accordingly.

Adding a suffix to $\mathcal{T}$ requires constant amortized time \cite{Ukkonen1995}. During the re-entrances to the expanding state, a chain of nodes was formed in $\mathcal{L}$ which was finally attached to the existing node in constant time when the buffering state was entered. For updating $\mathcal{A}$, attaching a chain of nodes to a tree requires linear time in the length of a chain \cite{Cole2005}. By amortizing all expanding calls, adding a new character takes amortized constant time.

To remove the oldest stored suffix from $\mathcal{T}$, we first find the corresponding leaf (e.g.\ by following a linked list of all leaves). If the leaf's parent has three or more children, the parent remains unchanged and we just remove the leaf from $\mathcal{T}$. Since leaves of $\mathcal{T}$ are not present in $\mathcal{L}$, $\mathcal{L}$ and consequently $\mathcal{A}$ remain unchanged.

On the other hand, if the leaf's parent has exactly two children, we remove the leaf from $\mathcal{T}$ and also its parent $\gamma_\mathcal{T}$ from $\mathcal{T}$ and $\gamma_\mathcal{L}$ from $\mathcal{L}$. To remove $\gamma_\mathcal{T}$ we merge its incoming and the remaining outgoing edges. Due to the following lemma, we can also safely remove $\gamma_\mathcal{L}$ since it is always a leaf in $\mathcal{L}$.

\begin{lemma}
	Let $\gamma_\mathcal{T}$ be a node with two children in $\mathcal{T}$, where one child is a leaf storing a position of the longest suffix $n-|W|$. Then, $\gamma_\mathcal{T}$ is not a terminating node of any suffix link.
\end{lemma}

\begin{proof}[Proof by contradiction.]
	Assume there is a node $\gamma'_\mathcal{T}$ in $\mathcal{T}$ with a suffix link pointing to $\gamma_\mathcal{T}$. Since $\gamma_\mathcal{T}$ has two children, $\gamma'_\mathcal{T}$ has at most two children, because $\gamma'_\mathcal{T}$ contains a subset of nodes of $\gamma_\mathcal{T}$. Observe the child of $\gamma_\mathcal{T}$ storing the position $n-|W|$ i.e.\ it spells out $W$. One child in $\gamma'_\mathcal{T}$ should then spell out $W$ prepended by some character. Since $W$ is already the longest suffix which exists in the window, a longer suffix and its corresponding leaf do not exist. Then, only one child of $\gamma'_\mathcal{T}$ remains and due to path compression $\gamma'_\mathcal{T}$ does not exist in $\mathcal{T}$ which contradicts the initial assumption.
\end{proof}

In the moment of removal, the removed leaf or its parent can be an active node $\beta$. If this is the case, then $B$ was a prefix of the removed suffix. Recall that at any time, $B$ corresponds to the longest repeated suffix of the window. Since the oldest suffix is removed by shifting the window, a new longest repeated suffix is consequently shortened for one character by updating $B$ to $B[2:]$. To find a new $\beta$ and an edge corresponding to the updated $B$, we simply follow the suffix link of the $\beta$'s parent and navigate the remainder of $B$ from the obtained node. The navigation time is amortized over all expanding calls, so finding a new $\beta$ requires amortized constant time.

To remove a leaf from $\mathcal{L}$ and $\mathcal{A}$ we require constant time in the worst case \cite{Cole2005}.

During the shift operation, no additional data structures are used. Consequently, the space complexity of the sliding suffix tree remains asymptotically unchanged. We conclude with the following theorem.

\begin{theorem}
	The sliding suffix tree of size $O(|W|)$ can be shifted in amortized $\Theta(1)$ time.
\end{theorem}

\section{Conclusions and Open Problems}
\label{sec:conclusions}

In this paper we presented a sliding suffix tree for performing online substring queries on a stream. By extending Ukkonen's online suffix tree construction algorithm, the presented data structure supports queries in optimal $\Theta(m+occ)$ time for alphabets of constant size while maintaining amortized constant time updates, where $m$ is the length of the query string and $occ$ the number of occurrences.

An open question remains whether the data structure can be updated in worst case constant time. There is a well known linear time suffix sorting lower bound \cite{Farach2000}, but to our knowledge, no per-character lower bound has been explored. Ukkonen's algorithm requires, by design, an amortized constant time for updates due to the implicit buffer of unfinalized nodes. To the best of our knowledge, no other online suffix tree construction algorithm has been developed without the implicit buffer.

In this paper, we assumed a constant size of the alphabet $\Sigma$ in asymptotic times for queries and updates. For arbitrary size of $\Sigma$, the current implementation of $\mathcal{T}$ data structure requires an additional factor of $\lg |\Sigma|$ time to determine a child at each step and maintain the same space complexity whereas $\mathcal{L}$ and $\mathcal{A}$ data structures are oblivious to $|\Sigma|$. An interesting question is whether the same asymptotic times can be achieved for integer alphabets as was done in \cite{Farach2000} for texts of fixed length. In our case $|\Sigma|=O(|W|)$, but the alphabet can change in time.

Streaming algorithms are common in heavy throughput environments, therefore it seems feasible to involve parallelism. Recently, two methods were introduced for performing fine-grained parallel queries on suffix trees \cite{Jekovec2015,Christiansen2016}. Both methods perform queries on static data structures only and perhaps supporting the shift operation used by the sliding suffix tree might be feasible. From a more coarse-grained parallelism point of view, the current query and update operations must be executed atomically. An interesting design question is whether the data structure could be designed in a mutable way, so a query and an update can be performed simultaneously, if different parts of the data structure are involved.

Finally, the presented data structure, while theoretically feasible, should also be competitive in practice. From our point of view, the main issue with tree-based data structures used in the sliding suffix tree is space consumption. The majority of the size accounts for the auxiliary data structure used for constant time lowest common ancestor. Some work on practical lowest common ancestor data structures has already been done in \cite{Fischer2006}. We believe that once the data structure is succinctly implemented, it should present a viable alternative to existing solutions.




\bibliography{jekovec}

\begin{thebibliography}{10}

\bibitem{Boyer1977}
Robert~S. Boyer and J.~Strother Moore.
\newblock {A fast string searching algorithm}.
\newblock {\em Communications of the ACM}, 20(10):762--772, oct 1977.
\newblock URL: \url{http://dl.acm.org/citation.cfm?id=359842.359859}, \href
  {http://dx.doi.org/10.1145/359842.359859} {\path{doi:10.1145/359842.359859}}.

\bibitem{Christiansen2016}
Anders~Roy Christiansen and Martin Farach-Colton.
\newblock {Parallel Lookups in String Indexes}.
\newblock In Shunsuke Inenaga, Kunihiko Sadakane, and Tetsuya Sakai, editors,
  {\em String Processing and Information Retrieval: 23rd International
  Symposium, SPIRE 2016, Beppu, Japan, October 18-20, 2016, Proceedings}, pages
  61--67. Springer International Publishing, Cham, 2016.
\newblock URL: \url{https://doi.org/10.1007/978-3-319-46049-9{\_}6}, \href
  {http://dx.doi.org/10.1007/978-3-319-46049-9_6}
  {\path{doi:10.1007/978-3-319-46049-9_6}}.

\bibitem{Cole2005}
Richard Cole and Ramesh Hariharan.
\newblock {Dynamic LCA Queries on Trees}.
\newblock {\em SIAM Journal on Computing}, 34(4):894--923, 2005.
\newblock URL: \url{https://doi.org/10.1137/S0097539700370539}, \href
  {http://dx.doi.org/10.1137/S0097539700370539}
  {\path{doi:10.1137/S0097539700370539}}.

\bibitem{Farach2000}
Martin Farach-Colton, Paolo Ferragina, and S.~Muthukrishnan.
\newblock {On the sorting-complexity of suffix tree construction}.
\newblock {\em Journal of the ACM}, 47(6):987--1011, nov 2000.
\newblock URL: \url{http://dl.acm.org/citation.cfm?id=355541.355547}, \href
  {http://dx.doi.org/10.1145/355541.355547} {\path{doi:10.1145/355541.355547}}.

\bibitem{Ferragina2000}
Paolo Ferragina and Giovanni Manzini.
\newblock {Opportunistic data structures with applications}.
\newblock In {\em Foundations of Computer Science, 2000. 41st Annual Symposium
  on}, pages 390--398, 2000.
\newblock \href {http://dx.doi.org/10.1109/SFCS.2000.892127}
  {\path{doi:10.1109/SFCS.2000.892127}}.

\bibitem{Fischer2006}
Johannes Fischer and Volker Heun.
\newblock {Theoretical and Practical Improvements on the RMQ-Problem, with
  Applications to LCA and LCE}.
\newblock In Moshe Lewenstein and Gabriel Valiente, editors, {\em Combinatorial
  Pattern Matching}, pages 36--48, Berlin, Heidelberg, 2006. Springer Berlin
  Heidelberg.

\bibitem{Harel1984}
Dov Harel and Robert~Endre Tarjan.
\newblock {Fast Algorithms for Finding Nearest Common Ancestors}.
\newblock {\em SIAM Journal on Computing}, 13(2):338--355, 1984.
\newblock URL: \url{https://doi.org/10.1137/0213024}, \href
  {http://dx.doi.org/10.1137/0213024} {\path{doi:10.1137/0213024}}.

\bibitem{Jekovec2015}
Matev{\v{z}} Jekovec and Andrej Brodnik.
\newblock {Parallel Query in the Suffix Tree}.
\newblock Technical report, University of Ljubljana, Faculty of Computer and
  Information Science, sep 2015.
\newblock URL: \url{http://arxiv.org/abs/1509.06167}, \href
  {http://arxiv.org/abs/1509.06167} {\path{arXiv:1509.06167}}.

\bibitem{Knuth1977}
Donald Knuth, J~{Morris Jr.}, and V~Pratt.
\newblock {Fast Pattern Matching in Strings}.
\newblock {\em SIAM Journal on Computing}, 6(2):323--350, 1977.
\newblock URL: \url{http://dx.doi.org/10.1137/0206024}, \href
  {http://dx.doi.org/10.1137/0206024} {\path{doi:10.1137/0206024}}.

\bibitem{Schieber1988}
Baruch Schieber and Uzi Vishkin.
\newblock {On Finding Lowest Common Ancestors: Simplification and
  Parallelization}.
\newblock {\em SIAM Journal on Computing}, 17(6):1253--1262, 1988.
\newblock URL: \url{https://doi.org/10.1137/0217079}, \href
  {http://dx.doi.org/10.1137/0217079} {\path{doi:10.1137/0217079}}.

\bibitem{Ukkonen1995}
Esko Ukkonen.
\newblock {On-Line Construction of Suffix Trees}.
\newblock {\em Algorithmica}, 14(3):249--260, 1995.

\end{thebibliography}


\end{document}